\documentclass[a4paper,twoside]{amsart}
\usepackage{amsmath}
\usepackage{amsfonts}
\usepackage{amssymb}
\usepackage{amsthm}
\usepackage{newlfont}
\usepackage{graphicx}
\usepackage{amscd}
\usepackage{young}

\theoremstyle{plain}
\newtheorem{Th}{Theorem}[section]
\newtheorem{Cor}[Th]{Corollary}
\newtheorem{Lem}[Th]{Lemma}
\newtheorem{Prop}[Th]{Proposition}
\theoremstyle{definition}
\newtheorem{Def}{Definition}[section]
\newtheorem{Ex}{Example}[section]
\theoremstyle{remark}
\newtheorem*{Rem}{Remark}
\numberwithin{equation}{section}

\newcommand{\DD}{{\mathbb D}}

\newcommand{\ZZ}{{\mathbb Z}}
\newcommand{\VV}{{\mathbb V}}

\newcommand{\bphi}{\boldsymbol{\phi}}

\begin{document}

\title
{Non-commutative double-sided continued fractions}

\author{Adam Doliwa}

\address{A. Doliwa: Faculty of Mathematics and Computer Science, University of Warmia and Mazury in Olsztyn,
ul.~S{\l}oneczna~54, 10-710~Olsztyn, Poland} 
\email{doliwa@matman.uwm.edu.pl}
\urladdr{http://wmii.uwm.edu.pl/~doliwa/}

\date{}

\begin{abstract}
We study double-sided continued fractions whose coefficients are non-commuting symbols. 
We work within the formal approach of the Mal'cev--Neumann series and free division rings. We start with presenting the analogs of the standard results from the theory of continued fractions, including their (right and left) simple fractions decomposition, the Euler--Minding summation formulas, and the relations between nominators and denominators of the simple fraction decompositions. We also transfer to the non-commutative double-sided setting the standard description of the continued fractions in terms of $2\times 2$ matrices presenting also a weak version of the Serret theorem. The equivalence transformations between the double continued fractions are described, including also the transformation from generic such fractions to their simplest form. Then we give the description of the double-sided continued fractions within the theory of quasideterminants and we present the corresponding version of the $LR$ and $qd$-algorithms. We study also (strictly and ultimately)  periodic  double-sided non-commutative continued fractions and we give the corresponding version of the Euler theorem. Finally we present a weak version of the Galois theorem and we give its relation to the non-commutative KP map, recently studied in the theory of discrete integrable systems.

\end{abstract}
\maketitle

\section{Introduction}
\subsection{Motivation}
Continued fractions theory is one of most celebrated subjects in mathematics studied by Euclide, Euler, Lagrange, Galois, and many others~\cite{Brezinski,FlajoletValleeVardi}. For a contemporary presentation of the theory of continued fractions and of their applications, including their generalizaton to power series fields, see~\cite{Brezinski-CFPA,JonesThron,Khrushchev,LorentzenWaadeland,Schmidt}.

Non-commutative continued fractions in their simplest form have been considered first by Wedderburn~\cite{Wedderburn}. In their one-sided forms they were discussed by Wynn~\cite{Wynn} in connection with non-commutative orthogonal polynomials, the $qd$~algorithm, and many other non-commutative versions of topics known in the theory of classical (commutative) continued fractions. 
The convergence questions for continued fractions in complex Banach algebras  were considered in \cite{Fair,Fair-SIAM,Baumann,Hayden}. The convergence of double-sided non-commutative continued fractions (the type of continued fractions considered in the present paper) was studied in \cite{DenkRiederle}, where one can find also wast bibliography on applications of non-commutative continued fractions in mathematics and physics. 
Convergence of Pad\'{e} approximants in a non-commutative normed algebra by using theorems of convergence of non-commutative continued fractions was investigated in~\cite{Draux}.
Combinatorial aspects of the non-commutative continued fractions are reviewed in \cite{Flajolet}, see also~\cite{Viennot-OP}. 
In \cite{Quasideterminants-GR1,Quasideterminants} the basic theory of non-commutative continued fractions was discussed from the point of view of quasideterminants. 
An iterative process for finding roots of the quadratic equation in Banach space in terms of periodic continuous fractions is given in~\cite{McFarland}. 
Quadratic equation with operator coefficients and their solution in terms of  strictly $2$-periodic continued fractions were considered in \cite{BusbyFair}.

The relation of integrable discrete systems to continued fractions and related numerical methods (Pad\'{e} approximation, orthogonal polynomials, convergence acceleration algorithms, the $qd$ and $LR$~algorithms) is well known, see \cite{PGR} or Section~4 of \cite{IDS} and references therein. An important example in this context is provided by the work of Moser~\cite{Moser} on the finite Toda lattice~\cite{Toda-TL} for which the inverse problem was solved by Stieltjes' continued fractions.  In this article we show close connection of the non-commutative continued fractions to the very basic non-commutative integrable system, called in \cite{Nimmo} the non-Abelian Hirota-Miwa system. Notice that its original commutative bilinear form was introduced by Hirota~\cite{Hirota} under the name of discrete analogue of a generalized Toda equation. We present it in the form of the so called companion to the KP (because of relation to the Kadomtsev--Petviashvilii equation) map~\cite{Doliwa-YB}. 
A solution, in terms of continued fractions, of the initial value problem for discrete-time Toda equation on a half-infinite lattice was constructed in~\cite{Common}. Non-commutative discrete systems in relation to continued fractions were studied also in \cite{DiFrancesco}.

Non-commutative extensions of discrete integrable systems are of growing interest in mathematical physics~\cite{FWN-Capel,BobSur-nc,Nimmo,DoliwaNoumi}. They provide a unifying platform to more thorough understanding of integrable classical, quantum or statistical mechanical lattice systems. Moreover, the non-commutativity makes formulas more rigid and in a better way  reveals structures responsible for the integrability. In certain sense, this resembles the conceptual step in transition from the integrable differential equations to integrable discrete systems. 

Although our research was motivated by applications of the non-commutative continued fractions to integrable systems, we have found it useful to study their general properties, in particular to present the corresponding version of the most pertinent results of the classical theory.

\subsection{Plan of the paper}
The paper is constructed as follows. After defining the non-commutative double-sided continued fractions as in \cite{Flajolet,Viennot-OP} we present the corresponding analogs (which we were not able to find in the literature) of standard results from the theory of finite continued fractions, including their (right and left) simple fractions decomposition, the Euler--Minding summation formulas, and the relations between nominators and denominators of the simple fraction decompositions. In our presentation we follow~\cite{Wedderburn,Wynn} where analogous results have been given on the level of the non-commutative right/left and the simple continued fractions. 
Then, in Section~\ref{sec:CF-2x2} we transfer to the non-commutative double-sided setting the standard description of the continued fractions in terms of $2\times 2$ matrices presenting also the weak version (in the sense of giving only implication in one side) of the Serret theorem. The equivalence transformations between the double continued fractions, understood as equality of all their convergents, are derived in Section~\ref{sec:CF-equiv}. This extends the well known transformation of generalized continued fractions to the non-commutative double-sided setting. In there we present also the transformation from the generic such fraction to the simplest (reduced) form of non-commutative continued fractions of Wedderburn. 

After concluding the first part of the paper with basic definitions and properties of non-commutative continued fractions in Section~\ref{sec:CF-Q} we give their description  using the theory of quasideterminants. We specify to our needs the more general definition of non-commutative continued fractions proposed by Gelfand and Retakh~\cite{Quasideterminants-GR1}. 
In the next Section~\ref{sec:LR-QD} we present the corresponding version of the $LR$ and $qd$-algorithms. We apply the standard ideas of the subject to our non-commutative setting generalizing also the results of Wynn~\cite{Wynn} derived for one-sided continued fractions. This allows to find also, following the commutative case described in \cite{Hirota-1983}, the corresponding non-commutative discrete-time Toda lattice equation.  

Then in Section~\ref{sec:CF-per} we introduce and study (strictly and ultimately)  periodic  double-sided non-commutative continued fractions. In particular, we give the corresponding version of Euler's theorem relating periodic continued fractions and quadratic expressions. The final Section~\ref{sec:Galois} is devoted to a weak version of the Galois theorem which relates two strictly periodic solutions of the same quadratic equation. In doing that we construct, in alternative way to that described in~\cite{DoliwaNoumi} in terms of characteristic series of certain context-free languages, the companion to the non-commutative KP map. 

In the paper we consider continued fractions on the level of Mal'cev--Neumann series \cite{Cohn-alg} with the natural topology, as known for example in the theory of formal languages \cite{Sakarovitch}. An analogous approach to continued fractions in power series fields was described in~\cite{Schmidt}. Therefore we conclude the introduction Section with recalling relevant notions. To justify our motivation we give also necessary information on the companion to the non-commutative KP map~\cite{KNY-A,Doliwa-YB,DoliwaNoumi}.

\subsection{The Mal'cev--Neumann construction and rational series over free group}

The free group $\Gamma = \Gamma (a_1,a_2,\dots)$  with generators $a_1,a_2, \dots$, possibly an infinite list, can be regarded as the set of equivalence classes of finite words in the letters $a_i$ and their formal inverses $a_i^{-1}$, where the words are considered equivalent if one can pass from one to the other by removing (or inserting) consecutive letters of the form $a_ia_i^{-1}$ or $a_i^{-1}a_i$. The group operation is concatenation and the empty word represents the identity element.

By fixing a total order $<$ in $\Gamma$ compatible with its group structure one can consider well ordered subsets $ P \subset \Gamma$, i.e. every non-empty subset of $P$
admits a smallest element. Let $\Bbbk$ be a field (in our case the field of rationals $\mathbb{Q}$) , the set $\Bbbk((\Gamma,<))$ of Mal'cev--Neumann series it is the set of the series in $\Bbbk[[\Gamma]]$ whose supports
are well ordered. One can equip $\Bbbk((\Gamma,<))$ with a $\Bbbk$-algebra structure by defining the (Cauchy) product of two Mal'cev--Neumann  series.
Remarkably~\cite{Malcev,Neumann} this algebra is a division ring (skew field).
Moreover, the space of Mal'cev--Neumann  series is complete in the natural ultra-metric  topology \cite{Neumann} (for simplicity we assume that the field $\Bbbk$ is equipped with discrete topology). Here a sequence converges if and only if the coefficient of every monomial stabilizes

For a finite set $A=\{ a_1, a_2,\dots , a_n\}$ of generators of $\Gamma$, the smallest sub-division ring in $\Bbbk((\Gamma,<))$ containing the group algebra $\Bbbk[\Gamma]$ is called the division ring of fractions (rational expressions). It turns out that the result is independent of the particular order $<$ used in the construction, in the sense that it is isomorphic \cite{Lewin} to the  universal field of fractions
$\Bbbk < \!\! \!\! \! ( \, A > \!\! \!\! )$
called also free division ring or the free skew field~\cite{Cohn}.

\subsection{The companion to the KP map}
\label{sec:comp-KP}
Consider the following linear problem \cite{Nimmo}
\begin{equation} \label{eq:lin-dKP}
\bphi_{(i)} - \bphi_{(j)} =  \bphi U_{ij},  \qquad i \ne j \leq N,
\end{equation}
where $U_{ij}\colon\ZZ^N \to \DD$ are functions defined on $N$-dimensional integer lattice with values in a division ring $\DD$, and the wave function $\bphi\colon \ZZ^N \to \VV(\DD)$ takes values in a right vector space over $\DD$; here the subscripts in brackets denote shifts in corresponding discrete variables, i.e., $f_{(i)}(n_1, \dots , n_i, \dots, n_N) = f(n_1,\dots , n_i + 1, \dots ,n_N)$.

The compatibility conditions of \eqref{eq:lin-dKP} consist of equations 
\begin{equation} \label{eq:alg-comp-U}
U_{ij} + U_{ji} = 0, \qquad  U_{ij} + U_{jk} + U_{ki} = 0, \qquad 
U_{kj}U_{ki(j)} = U_{ki} U_{kj(i)}, \quad i,j,k \quad \text{distinct,}
\end{equation}
called in \cite{Nimmo} non-Abelian Hirota--Miwa system of equations. Indeed, for commutative $\DD$ it is possible to introduce the potential $\tau$ such that
\begin{equation} \label{eq:U-tau}
U_{ij} = \frac{\tau \tau_{(ij)}}{\tau_{(i)} \tau_{(j)}}, \qquad i< j.
\end{equation}
The remaining part of the system
\eqref{eq:alg-comp-U} reduces to Hirota's discrete KP equation \cite{Hirota} (originally called the discrete analogue of a generalized Toda equation)
\begin{equation} \label{eq:H-M}
\tau_{(i)}\tau_{(jk)} - \tau_{(j)}\tau_{(ik)} + \tau_{(k)}\tau_{(ij)} =0,
\qquad 1\leq i< j <k \leq N,
\end{equation}
whose pivotal role in the whole KP hierarchy was discovered by Miwa~\cite{Miwa}; see also~\cite{KNS-rev} for review of its appearance in the theory of solvable systems of statistical and quantum physics.

In order to perform periodic reduction of the non-commutative system \eqref{eq:alg-comp-U} let us fix one (the last one, for example) coordinate $k=n_N$ and define 
\begin{equation}
u_{i,k}(n_1,\dots , n_{N-1}) = U_{N\, i}(n_1, \dots n_{N-1},k), \quad i = 1,\dots , N-1.
\end{equation}
Then equations \eqref{eq:alg-comp-U} transform to the so called non-commutative difference KP hierarchy
\begin{equation} \label{eq:ncKP-1}
u_{j,k}u_{i,k(j)} = u_{i,k} u_{j,k(i)}, \qquad  
u_{i,k(j)} + u_{j,k+1} = u_{j,k(i)} + u_{i,k+1}
\end{equation}
introduced on the commutative level in~\cite{KNY-A}.
The corresponding transformation rule
\begin{equation} \label{eq:KP-u-solved}
u_{i,k(j)} = ( u_{i,k} - u_{j,k})^{-1} u_{i,k} ( u_{i,k+1} - u_{j,k+1}), \qquad i\neq j, 
\end{equation} 
called in \cite{Doliwa-YB} the non-commutative KP map, is multidimensionally consistent~\cite{Dol-Des-red}.

Let us impose periodicity constraint $k\in\ZZ_P$ in the distinguished variable and, following~\cite{KNY-A}, consider the corresponding companion map 
$r \colon (\boldsymbol{x}, \boldsymbol{y}) \mapsto 
(\boldsymbol{x}^\prime, \boldsymbol{y}^\prime)$
defined by solution of the equations
\begin{equation} \label{eq:xy-k}
x_k y_k = x_k^\prime y_k^\prime, \qquad y_k + x_{k+1} = y_k^\prime + x_{k+1}^\prime, \qquad k\in\ZZ_P,
\end{equation}
compare with \eqref{eq:ncKP-1}. 
The map is involutive  $r\circ r = \mathrm{id}$ and, as a consequence of the multidimensional consistency of the KP map, satisfies the braid relations \cite{KNY-A,Doliwa-YB,DoliwaNoumi}
\begin{equation} \label{eq:braid-r}
r_{2} \circ r_{1} \circ r_{2} = r_{1} \circ r_{2} \circ r_{1}, 
\end{equation}
where  $r_1 = r \times \mathrm{id}$ and $r_2 = \mathrm{id}\times r$. 

Already in~\cite{KNY-A}, where the commutative version of the map~\eqref{eq:xy-k} was studied, it was observed that the problem of finding the companion map reduces to the following refactorization problem 
	\begin{equation}
	\begin{split}
\left(  \begin{array}{ccccc}  
x_1 & 0  & \cdots & 0 & 1 \\
1    & x_{2} & 0 & \cdots & 0 \\
0  &	1 & x_3  & \ddots &  \vdots  \\
\vdots & \ddots & \ddots & \ddots &  0   \\
0 & \cdots & 0 & 1 & x_{P} 
\end{array} \right) 
\left(  \begin{array}{ccccc}  
y_1 & 0  & \cdots & 0 & 1 \\
1    & y_{2} & 0 & \cdots & 0 \\
0  &	1 & y_3  & \ddots &  \vdots  \\
\vdots & \ddots & \ddots & \ddots &  0   \\
0 & \cdots & 0 & 1 & y_{P} 
\end{array} \right) 
 =  \hskip2cm \\ \hskip2cm = 
\left(  \begin{array}{ccccc}  
x_1^\prime & 0  & \cdots & 0 & 1 \\
1    & x_{2}^\prime & 0 & \cdots & 0 \\
0  &	1 & x_3^\prime  & \ddots &  \vdots  \\
\vdots & \ddots & \ddots & \ddots &  0   \\
0 & \cdots & 0 & 1 & x_{P}^\prime 
\end{array} \right) 
\left(  \begin{array}{ccccc}  
y_1^\prime & 0  & \cdots & 0 & 1 \\
1    & y_{2}^\prime & 0 & \cdots & 0 \\
0  &	1 & y_3^\prime  & \ddots &  \vdots  \\
\vdots & \ddots & \ddots & \ddots &  0   \\
0 & \cdots & 0 & 1 & y_{P}^\prime 
\end{array} \right) 
\end{split}
\end{equation}	
We will elaborate this observation, which links the problem to the continued fractions in  non-commuting symbols, and especially to the non-commutative analogue of the Galois theorem given in Section~\ref{sec:Galois}.
\begin{Rem}
	Another refactorization problem, called the local Yang--Baxter equation~\cite{Maillet-Nijhoff,Kashaev-Korepanov-Sergeev}, allows to find 
solutions to the Zamolodchikov tetrahedron equation~\cite{Zamolodchikov}. In particular, such solutions connected to the non-Abelian Hirota--Miwa system~\eqref{eq:alg-comp-U}, but without imposing the periodicity constraint have been found  recently in~\cite{Doliwa-Kashaev}. 
\end{Rem}

Ruling out the identity mapping, one can define \cite{NY-RSK,DoliwaNoumi} intermediate functions $h_k$ defined by
\begin{equation} 
\label{eq:y-h}
y_{k}^\prime + h_k^{-1} = y_{k} .
\end{equation}
From equations~\eqref{eq:xy-k} it is clear that they allow to construct the companion map. It is not difficult to show~\cite{DoliwaNoumi} that the functions can be expressed as series
\begin{equation} \label{eq:hk-series}
h_k= y_k^{-1} \sum_{n\geq 0} [y^{-1}x]_k^n = y_k^{-1} \left(1 + y_{k-1}^{-1}  x_k +
y_{k-1}^{-1} y_{k-2}^{-1}  x_{k-1}  x_k + \dots \right),
\end{equation}
where we used the following
notation
\begin{equation*}
[y^{-1}x]_k^0 = 1, \qquad [y^{-1}x]_k^{n+1} = y_{k-1}^{-1}[y^{-1}x]_{k-1}^n x_k \qquad \text{for} \quad n\geq 0.
\end{equation*}

\section{Double-sided non-commutative continued fractions --- basic properties}
\begin{Def}
Given symbols $\{b_0, a_1, b_1, c_1, \dots, a_n, b_n, c_n \}$, the finite double-sided non-commutative continued fraction is the following rational expression
\begin{equation} \label{eq:CF}
C_n = \left[ \begin{array}{cccc}
    & \!\!\!a_1 & \cdots & a_n \!\!\! \\
 \!\!\! b_0 & \!\!\!b_1 & \cdots & b_n \!\!\!\\
    & \!\!\!c_1 & \cdots & c_n \!\!\!
    \end{array} \right] =
b_0 + a_1 ( b_1 +a_2 ( b_2 + \dots + a_{n-1} ( b_{n-1} + a_n b_{n}^{-1}c_n )^{-1} \dots )^{-1} c_2)^{-1} c_1 ,
\end{equation}
considered in the corresponding free division ring.
For the one-sided continued fractions \eqref{eq:CF} with unital $c$-symbols (left-sided continued fractions) we will skip the third row (getting this way the classical Lagrange notation)
\begin{equation} \label{eq:CF-c}
\left[ \begin{array}{cccc}
& \!\!\!a_1 & \cdots & a_n \!\!\! \\
\!\!\! b_0 & \!\!\!b_1 & \cdots & b_n \!\!\!
\end{array} \right] =
b_0 + a_1 ( b_1 +a_2 ( b_2 + \dots + a_{n-1} ( b_{n-1} + a_n b_{n}^{-1})^{-1} \dots )^{-1} )^{-1},
\end{equation}
and similarly, we skip the first row in the case of unital $a$-symbols (right-sided continued fractions)
\begin{equation} \label{eq:CF-a}
\left[ \begin{array}{cccc}
\!\!\! b_0 & \!\!\!b_1 & \cdots & b_n \!\!\!\\
& \!\!\!c_1 & \cdots & c_n \!\!\!
\end{array} \right] =
b_0 +  ( b_1 +( b_2 + \dots +  ( b_{n-1} +  b_{n}^{-1}c_n )^{-1} \dots )^{-1} c_2)^{-1} c_1 .
\end{equation}
When both $a$- and $c$-symbols are unital (simple continued fraction) we use the standard notation
\begin{equation} \label{eq:CF-ac}
\left[ \begin{array}{cccc}
\!\!\! b_0 , & \!\!\!b_1, &\!\! \cdots ,& b_n \!\!\!
\end{array} \right] =
b_0 + ( b_1 +( b_2 + \dots +  ( b_{n-1} +  b_{n}^{-1} )^{-1} \dots )^{-1} )^{-1}.
\end{equation}
\end{Def}
\begin{Rem}
In \cite{Wynn} such continued fractions were denoted by
\begin{equation*}
\mathrm{post} \left[ b_0 + \frac{a_1}{b_1 +} \, \frac{a_2}{b_2 +} \cdots \frac{a_n}{b_n} \right]  = \left[ \begin{array}{cccc}
& \!\!\!a_1 & \cdots & a_n \!\!\! \\
\!\!\! b_0 & \!\!\!b_1 & \cdots & b_n \!\!\!
\end{array} \right], \qquad
\mathrm{pre} \left[ b_0 + \frac{c_1}{b_1 +} \, \frac{c_2}{b_2 +} \cdots \frac{c_n}{b_n} \right]  = \left[ \begin{array}{cccc}
\!\!\! b_0 & \!\!\!b_1 & \cdots & b_n \!\!\!\\
& \!\!\!c_1 & \cdots & c_n \!\!\!
\end{array} \right].
\end{equation*}
The simple non-commutative continued fractions \eqref{eq:CF-ac} were studied first in~\cite{Wedderburn}.
\end{Rem}
Later on in Section~\ref{sec:CF-equiv} we show that also in the non-commutative case double-sided continued fractions can be brought, by a suitable transformation, to the simple form \eqref{eq:CF-ac}. However such transformation is highly non-local, and we prefer to present basic properties of the double-sided continued fractions starting from their initial form~\eqref{eq:CF}.
\begin{Rem}
On the level of Mal'cev--Neumann division ring one can consider infinite continued fractions. Then their finite truncations are usually called their convergents/approximants. 
\end{Rem}
\begin{Prop} \label{prop:CF-AB}
	The continued fraction \eqref{eq:CF} can be brought to the form of a right simple fraction 
	\begin{equation} \label{eq:CF-AB}
C_n = A_n B_n^{-1}
	\end{equation}
	 with the nominator $A_n$ and the denominator $B_n$ calculated from the recurrence
	\begin{align} \label{eq:CF-rec-A}
	A_{k+1} & = A_{k}c_{k+1}^{-1}b_{k+1} + A_{k-1} c_k^{-1}a_{k+1},\\
	\label{eq:CF-rec-B}
	B_{k+1} & = B_{k}c_{k+1}^{-1}b_{k+1} + B_{k-1} c_k^{-1}a_{k+1},
	\end{align}
	with initial conditions 
	\begin{equation} \label{eq:CF-rec-init}
	A_{-1} = 1, \qquad A_0 = b_0, \qquad
	B_{-1} = 0, \qquad B_0 = 1,
	\end{equation}
	where, by definition, $c_0 = 1$.
\end{Prop}
\begin{proof}
	\begin{equation*}
	C_1 = b_0 + a_1 b_1^{-1}c_1 = \left( b_0 c_1^{-1}b_1 + a_1 \right) b_1^{-1}c_1,
	\end{equation*}
	which gives
	\begin{align*}
	A_1 & = b_0 c_1^{-1}b_1 + a_1  = A_0 c_1^{-1}b_1 + A_{-1}c_0^{-1}a_1, \\
	B_1 & = c_1^{-1} b_1  = B_0 c_1^{-1}b_1 + B_{-1}c_0^{-1}a_1 .
	\end{align*}
	The continued fraction $C_{n+1}$ is computed from $C_n$ by replacing $b_n$ by $b_n + a_{n+1} b_{n+1}^{-1} c_{n+1}$. Inserting this substitution into the simple fraction expression for $C_n$ we obtain
	\begin{align*}
	C_{n+1} = A_{n} (b_n \to b_n + a_{n+1} b_{n+1}^{-1} c_{n+1}) & B_{n}^{-1} (b_n \to b_n + a_{n+1} b_{n+1}^{-1} c_{n+1}) = \\
	= \left( A_{n-1}c_{n}^{-1} (b_n + a_{n+1} b_{n+1}^{-1} c_{n+1}) + A_{n-2} c_{n-1}^{-1}a_{n}
	\right) & \left( B_{n-1}c_{n}^{-1} (b_n + a_{n+1} b_{n+1}^{-1} c_{n+1}) + B_{n-2} c_{n-1}^{-1}a_{n}
	\right)^{-1} = \\
	 = \left( A_{n-1}c_{n}^{-1} a_{n+1} +  
	 \left( A_{n-1} c_n^{-1} b_n +  A_{n-2} c_{n-1}^{-1}a_{n}\right) 
	 \right.& \left. c_{n+1}^{-1} b_{n+1} \right) \times  \\  
	 \times \left( B_{n-1}c_{n}^{-1}  \right.& \left. a_{n+1}  +  
	 ( B_{n-1} c_n^{-1} b_n +  B_{n-2} c_{n-1}^{-1}a_{n}) c_{n+1}^{-1} b_{n+1} \right)^{-1} = \\
	 = \left( A_{n-1}c_{n}^{-1} a_{n+1} +  
	 A_n c_{n+1}^{-1} b_{n+1} \right) &
	 \left( B_{n-1}c_{n}^{-1} a_{n+1}  +  
	 B_n c_{n+1}^{-1} b_{n+1} \right)^{-1} .
	 	\end{align*}
\end{proof}
\begin{Cor} \label{cor:CF-BA}
	Equivalently, there is a left simple fraction form 
	\begin{equation} \label{eq:CF-BA}
C_n =  \tilde{B}_n^{-1}\tilde{A}_n
	\end{equation} 
	of the continued fraction \eqref{eq:CF}, which can be calculated from the recurrence
	\begin{align} \label{eq:CF-rec-A-rev}
	\tilde{A}_{k+1} & = b_{k+1} a_{k+1}^{-1} \tilde{A}_{k} + c_{k+1} a_k^{-1} \tilde{A}_{k-1} ,\\
	\label{eq:CF-rec-B-rev}
	\tilde{B}_{k+1} & =  b_{k+1} a_{k+1}^{-1} \tilde{B}_{k} + c_{k+1} a_k^{-1} \tilde{B}_{k-1} ,
	\end{align}
	with initial conditions \eqref{eq:CF-rec-init}, and with $a_0 = 1$.
\end{Cor}
\begin{Rem}
	As one can see, the right simple fraction form \eqref{eq:CF-AB} is suited well to one-sided continued fraction with unital $c$-coefficients (the nominators and denominators are then polynomials in the $a$- and $b$- coefficients) while left simple fraction form \eqref{eq:CF-BA} is suited well for unital $a$-coefficients. We will show in Section~\ref{sec:CF-equiv} that one can freely transfer between various forms of the non-commutative continued fractions, and therefore to keep the freedom we study their most general double-sided form.
\end{Rem}
The next property, which we study along lines described in \cite{Wynn}, is the double-sided analogue of the Euler--Minding theorem, which allows to replace a given continued fraction by the corresponding series.
\begin{Prop}
	The successive convergents of the continued fraction \eqref{eq:CF} are equal to the partial sums
	\begin{equation} \label{eq:CF-exp}
	C_n = b_0 + \sum_{k=1}^n (-1)^{k-1} (a_1 b_1^{-1} a_2 B_2^{-1})(B_1 c_2^{-1} a_3 B_3^{-1}) \dots
	(B_{k-3} c_{k-2}^{-1} a_{k-1} B_{k-1}^{-1})(B_{k-2} c_{k-1}^{-1} a_k B_k^{-1}) .
	\end{equation}
\end{Prop}
\begin{proof}
	Eliminating $b_{k+1}$ from equations \eqref{eq:CF-rec-A}-\eqref{eq:CF-rec-B} we have
	\begin{equation} \label{eq:CF-AA-BB}
	A_{k}^{-1}A_{k+1} - B_k^{-1} B_{k+1} = (A_{k}^{-1}A_{k-1} - B_k^{-1} B_{k-1}) c_k^{-1} a_{k+1},
	\end{equation}
	which gives the following recurrence between successive convergents of the continued fraction
	\begin{align*}
	C_{k+1} - C_k & = - (C_k - C_{k-1}) b_{k-1} c_k^{-1} a_{k+1} B_{k+1}^{-1}  \\
	& = 
	(-1)^k (a_1 b_1^{-1} a_2 B_2^{-1})(B_1 c_2^{-1} a_3 B_3^{-1}) \dots
	(B_{k-1} c_k^{-1} a_{k+1} B_{k+1}^{-1}).
	\end{align*}
\end{proof}
\begin{Cor}
	By Corollary~\ref{cor:CF-BA} there exists an analogous sum with left simple-fraction denominators
	\begin{equation} \label{eq:CF-exp-rev}
	C_n = b_0 + \sum_{k=1}^n (-1)^{k-1} 
	(\tilde{B}_{k}^{-1} c_{k} a_{k-1}^{-1} \tilde{B}_{k-2}) 
	(\tilde{B}_{k-1}^{-1} c_{k-1} a_{k-2}^{-1} \tilde{B}_{k-3}) 
	 \dots (\tilde{B}_{3}^{-1} c_{3} a_{2}^{-1} \tilde{B}_{1}) 
	 (\tilde{B}_{2}^{-1} c_2 b_1^{-1}c_1).
		\end{equation}
\end{Cor}

\begin{Ex}
	For $n=2$ we have the continued fraction 
	\begin{equation*}
	C_2  = b_0 + a_1 ( b_1 + a_2 b_2^{-1}c_2)^{-1}c_1,
	\end{equation*}
	with nominator and denominator of the corresponding right simple fraction of the form
	\begin{align*}
	A_2 & = (b_0 c_1^{-1} b_1 + a_1) c_2^{-1} b_2 + b_0 c_1^{-1} a_2 ,\\
	B_2 & = c_1^{-1} b_1 c_2^{-1} b_2 + c_1^{-1} a_2 .
	\end{align*}
	The nominator and denominator of the left simple fraction are given by
	\begin{align*}
	\tilde{A}_2 & =b_2 a_2^{-1} (b_1 a_1^{-1} b_0 + c_1)   + c_2 a_1^{-1} b_0 ,\\
	\tilde{B}_2 & =b_2 a_2^{-1} b_1 a_1^{-1} + c_2 a_1^{-1}  .
	\end{align*}
	As the  sum, whose truncations  give subsequent convergents, it reads
	\begin{equation*}
	C_2 = b_0 + a_1 b_1^{-1}c_1 - a_1 b_1^{-1}a_2 (c_1^{-1} b_1 c_2^{-1} b_2 + c_1^{-1} a_2)^{-1}.
	\end{equation*}
\end{Ex}
We close presentation of the simplest properties of non-commutative double-sided fractions by giving relations between the nominators and denominators of their simple fraction decompositions. We follow the line of the corresponding research in~\cite{Wedderburn} applied there to simple (with both $a$- and $c$-coefficients unital) continued fractions. Notice symmetry in their formulation due to presence of both coefficients.
\begin{Cor}
	Apart from the obvious Ore-type relation
	\begin{equation} \label{eq:CF-Ore}
	\tilde{B}_n A_n = \tilde{A}_n B_n,
	\end{equation}
	 between the nominators and denominators of the two different decompositions of the continued fraction in simple fractions, 	
	the following other relations hold
	\begin{align} \label{eq:CF-Delta}
	\tilde{B}_{n-1} A_n -\tilde{A}_{n-1} B_n = (-1)^{n-1}a_n, \qquad &
	\tilde{B}_n A_{n-1} - \tilde{A}_n B_{n-1}= (-1)^n c_n, \\ \label{eq:CF-Delta-2}
    A_n a_n^{-1} \tilde{B}_{n-1} - A_{n-1} c_n^{-1} \tilde{B}_n = (-1)^{n+1}, \qquad	&
    B_n a_n^{-1} \tilde{A}_{n-1} - B_{n-1} c_n^{-1} \tilde{A}_n = (-1)^{n},\\
    \label{eq:CF-Delta-3}
    A_n a_n^{-1} \tilde{A}_{n-1} = A_{n-1} c_n^{-1} \tilde{A}_n, \qquad &
    B_n a_n^{-1} \tilde{B}_{n-1} = B_{n-1} c_n^{-1} \tilde{B}_n.
	\end{align}
\end{Cor}
\begin{proof}
	To prove equations \eqref{eq:CF-Delta} let
	\begin{equation*}
	\Delta_n = \tilde{B}_{n-1} A_n -\tilde{A}_{n-1} B_n , \qquad 
	\tilde{\Delta}_n = \tilde{B}_n A_{n-1} - \tilde{A}_n B_{n-1}.
	\end{equation*}
Then by recurrence relations \eqref{eq:CF-rec-A}-\eqref{eq:CF-rec-B} and
\eqref{eq:CF-rec-A-rev}-\eqref{eq:CF-rec-B-rev}, and equation \eqref{eq:CF-Ore} one obtains 
\begin{equation*}
\Delta_{n+1} = \tilde{\Delta}_n c_n^{-1} a_{n+1}, \qquad \tilde{\Delta}_{n+1} = c_{n+1} a_n^{-1}\Delta_n .
\end{equation*}
Directly one can show that 
$\Delta_1 = a_1$ and $\tilde{\Delta}_1 = -c_1$, which concludes the first part of the proof.

To show the first equation of \eqref{eq:CF-Delta-2} let 
\begin{equation*}
\Delta^\prime_n = A_n a_n^{-1} \tilde{B}_{n-1} - A_{n-1} c_n^{-1} \tilde{B}_n .
\end{equation*}
Then the statement follows from the fact that
\begin{equation*}
\Delta^\prime_{n+1} = - \Delta^{\prime}_n, \quad \text{and} \quad \Delta^\prime_1 = 1.
\end{equation*}
The second equation of \eqref{eq:CF-Delta-2} can be shown exactly in the same fashion. 

To show the first equation of \eqref{eq:CF-Delta-3} let 
\begin{equation*}
\Delta^{\prime\prime}_n = A_n a_n^{-1} \tilde{A}_{n-1} - A_{n-1} c_n^{-1} \tilde{A}_n .
\end{equation*}
Then the statement is a simple consequence the fact that
\begin{equation*}
\Delta^{\prime\prime}_{n+1} = - \Delta^{\prime\prime}_n, \quad \text{and} \quad \Delta^{\prime\prime}_1 = 0.
\end{equation*}
The second equation of \eqref{eq:CF-Delta-3} can be shown analogously. 
\end{proof}
	
\section{Continued fractions in terms of $2\times 2$ matrices}
\label{sec:CF-2x2}
Calculation of the continued fraction $C_n$ can be split into subsequent calculation of $Y_n=b_n$, $Y_{n-1},\dots $, $Y_1, Y_0 = C_n$, where  
\begin{equation} \label{eq:CF-system-Y}
Y_{k-1} = b_{k-1} + a_k Y_{k}^{-1} c_k, \quad k=n, \dots ,1 .
\end{equation}
The system can be put in the form
\begin{equation*}
\left( \begin{array}{c}
1 \\ Y_{k-1} \end{array} \right) c_k^{-1} Y_{k}  = 
\left( \begin{array}{cc}
0 & c_k^{-1}  \\ a_{k} &   b_{k-1}c_k^{-1} \end{array} \right) 	
\left( \begin{array}{c}
1 \\ Y_{k} \end{array} \right) , \qquad k=1,\dots , n-1 .
\end{equation*}
By induction one shows the following result.
\begin{Prop}
In terms of $2\times 2$ matrices the recurrence relations \eqref{eq:CF-rec-A}-\eqref{eq:CF-rec-B} take the form
\begin{equation} \label{eq:CF-2m2-ind}
\left( \begin{array}{cc}
0 & c_1^{-1} \\ a_{1} &   b_0 c_1^{-1} \end{array} \right) 	
\left( \begin{array}{cc}
0 & c_2^{-1} \\  a_{2} &   b_1 c_2^{-1} \end{array} \right) 	\cdots
\left( \begin{array}{cc}
0 & c_{k+1}^{-1} \\  a_{k+1} &  b_k c_{k+1}^{-1}\end{array} \right) =
\left( \begin{array}{cc}
B_{k-1}c_k^{-1} a_{k+1}  & B_k c_{k+1}^{-1}\\ A_{k-1}c_k^{-1} a_{k+1} & A_k c_{k+1}^{-1} \end{array} \right) ,
\end{equation}
what implies
\begin{equation} \label{eq:CF-2m2-AB}
\left( \begin{array}{c}
1 \\ C_n \end{array} \right) c_1^{-1} Y_{1} c_2^{-1} Y_2 \dots c_n^{-1} Y_n = 
\left( \begin{array}{cc}
B_{n-2}c_{n-1}^{-1} a_{n}  & B_{n-1} c_n^{-1} \\ A_{n-2}c_{n-1}^{-1} a_{n} & A_{n-1} c_n^{-1} \end{array} \right) \left( \begin{array}{c}
1 \\ b_n \end{array} \right) = \left( \begin{array}{c}
B_n \\ A_n \end{array} \right)
\end{equation}
and gives decomposition \eqref{eq:CF-AB}.
\end{Prop}
\begin{Rem}
	In the case of \emph{commutative coefficients} of the continued fractions, by taking determinant of both sides of the recurrence relation \eqref{eq:CF-2m2-ind} we would obtain the so called determinant formula between numerators and denominators of two successive approximants. In the fully non-commutative case one cannot expect such a formula due to the lack of appropriate notion of the determinant in that case; notice however equation \eqref{eq:CF-AA-BB}. The relation of non-commutative continued fractions to quasideterminants is presented in Section~\ref{sec:CF-Q}. 
\end{Rem}
\begin{Cor}
	The  system \eqref{eq:CF-system-Y} can be put in equivalent form
	\begin{equation*}
	Y_k a_k^{-1} \;  \left( \begin{array}{cc}
	1 \; , & \!\! Y_{k-1} \end{array} \right)  = \left( \begin{array}{cc}
	1 \; , & \!\! Y_{k} \end{array} \right)
	\left( \begin{array}{cc}
	0 &c_{k}   \\ a_k^{-1} &a_k^{-1}   b_{k-1}  \end{array} \right) 	
	, \qquad k= 1,\dots , n-1,
	\end{equation*}
	what leads to
	\begin{equation} \label{eq:CF-2m2-AB-tilde}
	Y_n a_n^{-1}   \dots Y_1 a_1^{-1} \; \left( \begin{array}{cc}
	1 \; , & \!\! C_n \end{array} \right)  = \left( \begin{array}{cc}
	1 \; , & \!\! b_n  \end{array} \right)
	\left( \begin{array}{cc}
	c_n a_{n-1}^{-1} \tilde{B}_{n-2} & c_n a_{n-1}^{-1} \tilde{A}_{n-2}   \\ 
	a_n^{-1}\tilde{B}_{n-1} & a_n^{-1}\tilde{A}_{n-1} \end{array} \right) 	= \left( \begin{array}{cc}
	\tilde{B}_{n} \; , & \!\! \tilde{A}_{n} \end{array} \right),
	\end{equation}
	and gives decomposition \eqref{eq:CF-BA}.
\end{Cor}
\begin{Rem}
In addition to $Y_k$, $k=0,\dots , n$ let us introduce $X_k$, $k=1,\dots , n$, by 
$X_k= a_k Y_k^{-1}$. Then the system \eqref{eq:CF-system-Y} takes the form
\begin{align} \label{eq:CF-system-XY-a}
X_k Y_k & = a_k ,\\ \label{eq:CF-system-XY-b}
Y_{k-1} - X_{k} c_{k} & = b_{k-1} ,
\end{align}
where $k=1,\dots , n$, and the initial data is $Y_n = b_n$. We will come back to equations \eqref{eq:CF-system-XY-a}-\eqref{eq:CF-system-XY-b} in Section~\ref{sec:CF-Q}.
\end{Rem}

The following consequence of decomposition \eqref{eq:CF-2m2-AB} will be used in Section~\ref{sec:CF-per}.
\begin{Cor} \label{cor:CF-splitting}
	Assume that the data of the continued fraction split into two parts related to the partition $n=K+m$
\begin{equation} C_{K+m} =
\left[ \begin{array}{cccccc}
& \!\!\!a_1 & \cdots & a_K & \cdots & a_{K+m} \!\!\! \\
\!\!\! b_0 & \!\!\!b_1 & \cdots & b_K & \cdots &b_{K+m} \!\!\!\\
& \!\!\!c_1 & \cdots & c_K & \cdots &c_{K+m} \!\!\!
\end{array} \right] = 
\left[ \begin{array}{cccc}
& \!\!\!a_1 & \cdots & a_K \!\!\! \\
\!\!\! b_0 & \!\!\!b_1 & \cdots & C^K_m \!\!\!\\
& \!\!\!c_1 & \cdots & c_K \!\!\!
\end{array} \right],
\end{equation}	
where	
\begin{equation} 
C^K_m =
\left[ \begin{array}{cccc}
& \!\!\!a_{K+1} & \cdots & a_{K+m} \!\!\! \\
\!\!\! b_K & \!\!\!b_{K+1} & \cdots & b_{K+m} \!\!\!\\
& \!\!\!c_{K+1} & \cdots & c_{K+m} \!\!\!
\end{array} \right] .
\end{equation}	
By $A^K_j$ and $B^K_j$ denote (right) nominators and denominators of the convergents $C^K_j$ of the shifted continued fraction $C^K_m$, then equation \eqref{eq:CF-2m2-ind} implies
\begin{equation*}
\left( \begin{array}{cc}
B_{K-2}c_{K-1}^{-1} a_{K}  & B_{K-1} c_K^{-1} \\ A_{K-2}c_{K-1}^{-1} a_{K} & A_{K-1} c_K^{-1} \end{array} \right) 
\left( \begin{array}{cc}
B^K_{K+m-2}  & B^K_{K+m-1} \\ A^K_{K+m-2} & A^K_{K+m-1}  \end{array} \right) = 
  \left( \begin{array}{cc}
B_{K+m-2}  & B_{K+m-1}  \\ A_{K+m-2} & A_{K+m-1} \end{array} \right).
\end{equation*}
\end{Cor}

We close the Section with a weak version of the Serret theorem. The adjective "weak" in this paper means that we present only the implication in one direction, starting from a given form of the continued fraction. 
\begin{Cor} \label{cor-CF-Serret}
	Given two continued fractions $C$ and $\bar{C}$ whose coefficients are ultimately equal, i.e. for certain $K,L \geq 0$, and all $k\geq 0$
	\begin{equation*}
	a_{K+ k + 1} = \bar{a}_{L+k+1},\qquad 
	b_{K+ k } = \bar{b}_{L+k},\qquad 
	c_{K+ k + 1} = \bar{c}_{L+k+1},
	\end{equation*}
	then they are linked by right and left fractional-linear transformations
	\begin{equation}
\label{eq:CF-Serret}
\bar{C} = (\alpha + \beta C)(\gamma + \delta C)^{-1} = (\tilde{\alpha} + C \tilde{\beta} )^{-1} (\tilde{\gamma} + C\tilde{\delta}).
	\end{equation}
\end{Cor}
\begin{proof}
	Denote by $C^*$ the common part of the continued fractions
		\begin{equation*}
		C^* =  \left[ \begin{array}{cccc}
		& \!\!\!a_{K+1} & a_{K+2} &\cdots    \!\!\! \\
		\!\!\! b_K & \!\!\!b_{K+1} & b_{K+2} & \cdots   \!\!\!\\
		& \!\!\!c_{K+1} & c_{K+2} & \cdots  \!\!\! 
		\end{array} \right] =
		  \left[ \begin{array}{cccc}
		& \!\!\!\bar{a}_{L+1} & \bar{a}_{L+2} &\cdots    \!\!\! \\
		\!\!\! \bar{b}_{L} & \!\!\!\bar{b}_{L+1} & \bar{b}_{L+2} & \cdots   \!\!\!\\
		& \!\!\!\bar{c}_{L+1} & \bar{c}_{L+2} & \cdots  \!\!\! 
		\end{array} \right] .
		\end{equation*}
	Since by \eqref{eq:CF-BA}, with $C^*$ in the place of $b_K$, we have
	\begin{equation*}
	C =  \left[ \begin{array}{ccccc}
	& \!\!\!a_1 & \cdots & a_{K-1} & a_K \!\!\! \\
	\!\!\! b_0 & \!\!\!b_1 & \cdots & b_{K-1} & C^* \!\!\!\\
	& \!\!\!c_1 & \cdots & c_{K-1} & c_K \!\!\! 
	\end{array} \right] 
	=
	( c_K a_{K-1}^{-1} \tilde{B}_{K-2} + C^* 	a_K^{-1}\tilde{B}_{K-1} )^{-1}
	( c_K a_{K-1}^{-1} \tilde{A}_{K-2} + C^* 	a_K^{-1}\tilde{A}_{K-1} ) ,
	\end{equation*}
	then its inversion gives
	\begin{equation} \label{eq:CF-C*C}
	C^* = c_K a_{K-1}^{-1} ( \tilde{A}_{K-2} - \tilde{B}_{K-2} C)
	(\tilde{B}_{K-1} C - \tilde{A}_{K-1} )^{-1}  a_K .
	\end{equation}
	Inserting the above expression into the result of similar application of decomposition \eqref{eq:CF-AB}
		\begin{equation*}
		\bar{C} =  \left[ \begin{array}{ccccc}
		& \!\!\!\bar{a}_1 & \cdots & \bar{a}_{L-1} & \bar{a}_L \!\!\! \\
		\!\!\! \bar{b}_0 & \!\!\!\bar{b}_1 & \cdots & \bar{b}_{L-1} & C^* \!\!\!\\
		& \!\!\!\bar{c}_1 & \cdots & \bar{c}_{L-1} & \bar{c}_L \!\!\! 
		\end{array} \right] =
( \bar{A}_{L-2}\bar{c}_{L-1}^{-1} \bar{a}_{L}  + \bar{A}_{L-1} \bar{c}_L^{-1} C^* )
( \bar{B}_{L-2}\bar{c}_{L-1}^{-1} \bar{a}_{L}  + \bar{B}_{L-1} \bar{c}_L^{-1} C^* )^{-1}
		\end{equation*}
		we obtain the first of equations \eqref{eq:CF-Serret}. The second one can be obtained analogously by exchanging the role of equations~\eqref{eq:CF-AB} and \eqref{eq:CF-BA}.
\end{proof}

\section{The equivalence transformation}
\label{sec:CF-equiv}
\begin{Def}
	Two continued fractions are called equivalent if they
	have the same sequence of convergents.
\end{Def}

\begin{Prop}
Given sequence of non-zero elements $(\gamma_k , \delta_k)_{k\geq 0}$ with $\gamma_0 = \delta_0 = 1$, define transformation of the coefficients of the continued fraction~\ref{eq:CF} by
\begin{equation} \label{eq:CF-equiv}
(a_k , b_k , c_k ) \to (\bar{a}_k,\bar{b}_k, \bar{c}_k) = 
(\gamma_{k-1} a_k \delta_k, \gamma_{k} b_k \delta_k, \gamma_{k} c_k \delta_{k-1}), \qquad k\geq 1, \qquad \bar{b}_0 = b_0.
\end{equation}
Then 
\begin{equation} \label{eq:CF-equiv}
\left[ \begin{array}{cccc}
& \!\!\!\bar{a}_1 & \cdots & \bar{a}_n \!\!\! \\
\!\!\! \bar{b}_0 & \!\!\!\bar{b}_1 & \cdots & \bar{b}_n \!\!\!\\
& \!\!\!\bar{c}_1 & \cdots & \bar{c}_n \!\!\!
\end{array} \right] =
\left[ \begin{array}{cccc}
& \!\!\!a_1 & \cdots & a_n \!\!\! \\
\!\!\! b_0 & \!\!\!b_1 & \cdots & b_n \!\!\!\\
& \!\!\!c_1 & \cdots & c_n \!\!\!
\end{array} \right],
\end{equation}
in particular, the corresponding successive numerators $\bar{A}_k$ and denominators $\bar{B}_k$, as described in Proposition~\ref{prop:CF-AB}, are given by
\begin{equation} \label{eq:CF-equiv-AB}
\bar{A}_k = A_k \delta_k, \qquad \bar{B}_k = B_k \delta_k.
\end{equation}
Similarly, the nominators $\tilde{\bar{A}}_k$ and the denominators $\tilde{\bar{B}}_k$ of the corresponding simple fraction decomposition, as described in Corollary~\ref{cor:CF-BA}, are given by
\begin{equation} \label{eq:CF-equiv-BA}
\tilde{\bar{A}}_k = \gamma_k \tilde{A}_k , \qquad \tilde{\bar{B}}_k = \gamma_k \tilde{B}_k.
\end{equation}	
\end{Prop}
\begin{proof}	The first nominator $\bar{A}_1$ can be calculated directly (recall that $\bar{A}_0 = b_0 = A_0 \delta_0$)
	\begin{equation*}
	\bar{A}_1 = \bar{b}_0 \bar{c}_1^{-1} \bar{b}_1 + \bar{a}_1 = (b_0 c_1^{-1} b_1 + a_1 ) \delta_1 = A_1 \delta_1,
	\end{equation*}
and the formulas for subsequent nominators follow by induction
\begin{equation*}
\bar{A}_{k+1} = \bar{A}_{k} \bar{c}_{k+1}^{-1} \bar{b}_{k+1}  + 
\bar{A}_{k-1} \bar{c}_{k}^{-1} \bar{a}_{k+1} =
(A_{k} c_{k+1}^{-1} b_{k+1} + A_{k-1} c_k^{-1} a_{k+1})\delta_{k+1}.
\end{equation*}
The denominators $\bar{B}_k$ are treated in the same way, which concludes proof of formulas \eqref{eq:CF-equiv} and \eqref{eq:CF-equiv-AB}.
Equations \eqref{eq:CF-equiv-BA} can be proven analogously.
\end{proof}
\begin{Ex}
	(i) The sequence $(\gamma_k,\delta_k) = (c_{k}^{-1},1)$ removes the $c$-symbols 
	\begin{equation}
\left[ \begin{array}{ccccc}
& \!\!\!a_1 & a_2 &\cdots & a_n \!\!\! \\
\!\!\! b_0 & \!\!\!b_1 & b_2 &\cdots & b_n \!\!\!\\
& \!\!\!c_1 & c_2 &\cdots & c_n \!\!\!
\end{array} \right] = 	
\left[ \begin{array}{ccccc}
& \!\!\!a_1 & c_1^{-1} a_2 &\cdots & c_{n-1}^{-1} a_n \!\!\! \\
\!\!\! b_0 & \!\!\! c_1^{-1} b_1 & c_2^{-1} b_2 &\cdots & c_n^{-1}b_n \!\!\!
\end{array} \right] ,	
	\end{equation}
	while (ii) the sequence $(\gamma_k,\delta_k) = (1,a_{k}^{-1})$ removes the $a$-symbols 
	\begin{equation}
	\left[ \begin{array}{ccccc}
	& \!\!\!a_1 & a_2 &\cdots & a_n \!\!\! \\
	\!\!\! b_0 & \!\!\!b_1 & b_2 &\cdots & b_n \!\!\!\\
	& \!\!\!c_1 & c_2 &\cdots & c_n \!\!\!
	\end{array} \right] = 	
	\left[ \begin{array}{ccccc}
	\!\!\! b_0 & \!\!\!  b_1 a_1^{-1} &  b_2 a_2^{-1} &\cdots & b_n   a_n^{-1}\!\!\! \\
	& \!\!\!c_1 & c_2 a_1^{-1} &\cdots & c_n a_{n-1}^{-1} \!\!\! 
	\end{array} \right] .	
	\end{equation}
\end{Ex}
\begin{Ex}
	(i) When the transformation with the sequence 
	\begin{equation*}
	(\gamma_1,\delta_1) = (b_1^{-1},a_1^{-1}), \qquad 
	(\gamma_2,\delta_2) = (b_2^{-1},a_2^{-1} b_1 ), \quad \text{and} \quad  
	(\gamma_k,\delta_k) = (b_{k}^{-1},a_k^{-1}b_{k-1})  \quad \text{for} \quad k\geq 3, 
	\end{equation*}
	is applied to the continued fraction with unital $c$-symbols then it gives 
	the continued fraction with unital $a$-symbols
	\begin{equation}
	\left[ \begin{array}{cccccc}
	& \!\!\!a_1 & a_2 & a_3 & \cdots & a_n \!\!\! \\
	\!\!\! b_0 & \!\!\!b_1 & b_2 & b_3 & \cdots & b_n \!\!\!
	\end{array} \right] =	
		\left[ \begin{array}{cccccc}
		\!\!\! b_0 & \!\!\! a_1^{-1} &  a_2^{-1} b_ 1 & a_3^{-1} b_2 & \cdots &    a_n^{-1} b_{n-1} \!\!\! \\
		& \!\!\! b_1^{-1} & b_2^{-1} a_1^{-1} & b_3^{-1} a_2^{-1} b_1 & \cdots & b_n^{-1} a_{n-1}^{-1} b_{n-2}\!\!\! 
		\end{array} \right] ,
	\end{equation}
	while (ii) application of the sequence
	\begin{equation*}
	(\gamma_1,\delta_1) = (c_1^{-1},b_1^{-1}), \qquad 
	(\gamma_2,\delta_2) = (b_1 c_2^{-1} , b_2^{-1}), \quad \text{and} \quad  
	(\gamma_k,\delta_k) = (b_{k-1} c_k^{-1},b_k^{-1})  \quad \text{for} \quad k\geq 3, 
	\end{equation*}
	to the continued fraction with unital $a$-symbols  gives
	the continued fraction with unital $c$-symbols
	\begin{equation}\left[ \begin{array}{cccccc}
	\!\!\! b_0 & \!\!\!b_1 & b_2 & b_3 & \cdots & b_n \!\!\! \\
		& \!\!\!c_1 & c_2 & c_3 & \cdots & c_n \!\!\! 
	\end{array} \right] =	
	\left[ \begin{array}{cccccc}
	& \!\!\! b_1^{-1} & c_1^{-1} b_2^{-1}  &b_1 c_2^{-1}  b_3^{-1}  & \cdots & b_{n-2} c_{n-1}^{-1}  b_n^{-1}  \!\!\!  \\
	\!\!\! b_0 & \!\!\! c_1^{-1} &   b_ 1 c_2^{-1}&  b_2 c_3^{-1}& \cdots &     b_{n-1} c_n^{-1}\!\!\! 	
	\end{array} \right] .
	\end{equation}
\end{Ex}

\begin{Prop} \label{prop:CF-red}
	The equivalence transformation with the sequences
	\begin{align*}
	\gamma_{2i+1} = c_1^{-1} a_2 c_3^{-1} \dots c_{2i-1}^{-1} a_{2i} c_{2i+1}^{-1}, & \qquad
	\gamma_{2i} = a_1 c_2^{-1} a_3 \dots a_{2i-1} c_{2i}^{-1} \\
	\delta_{2i+1} = a_{2i+1}^{-1} c_{2i} a_{2i-1}^{-1} \dots a_3^{-1} c_2 a_1^{-1}, & \qquad 
	\delta_{2i} = a_{2i}^{-1} c_{2i-1} \dots c_3 a_2^{-1} c_1 ,
	\end{align*}
	produces the continued fraction in reduced form (i.e. with unital $a$-coefficients and unital $c$-coefficients), and the new $b$-coefficients given by
	\begin{equation*}
	\bar{b}_{2i+1} = c_1^{-1}a_2 \dots a_{2i} c_{2i+1}^{-1} b_{2i+1}
	a_{2i+1}^{-1} c_{2i} \dots c_2 a_1^{-1} ,\quad 
	\bar{b}_{2i} = a_1 c_{2}^{-1} \dots a_{2i-1} c_{2i}^{-1} b_{2i} a_{2i}^{-1} c_{2i-1} \dots a_2^{-1} c_1 , 
	\end{equation*}
	i.e.
	\begin{equation}
	\left[ \begin{array}{ccccc}
	& \!\!\!a_1 & a_2 &\cdots & a_n \!\!\! \\
	\!\!\! b_0 & \!\!\!b_1 & b_2 &\cdots & b_n \!\!\!\\
	& \!\!\!c_1 & c_2 &\cdots & c_n \!\!\!
	\end{array} \right] = 	
	\left[ 
	 \bar{b}_0 , \bar{b}_1 , \bar{b}_2, \cdots ,\bar{b}_n 
	 \right] .	
	\end{equation}
\end{Prop}
\begin{proof}
	Conditions
	\begin{equation*}
	\gamma_{k-1} a_k \delta_k =1  \qquad \text{and} \qquad \gamma_k c_k \delta_{k-1} = 1
	\end{equation*}
	give recurrence relations which allow to calculate the sequences of the transformation.
\end{proof}
\begin{Cor} 
	Exactly in the same way as above one can show that:\\
	(i) the sequence 
	\begin{equation*}
	\gamma_k = c_1^{-1} b_1 \dots c_{k-1}^{-1} b_{k-1} c_k^{-1}, \qquad
	\delta_k = b_k^{-1} c_k b_{k-1}^{-1} \dots b_1^{-1} c_1 
	\end{equation*} removes both the $b$- and $c$-symbols 
	\begin{equation*}
	\left[ \begin{array}{ccccc}
	& \!\!\!a_1 & a_2 &\cdots & a_n \!\!\! \\
	\!\!\! b_0 & \!\!\!b_1 & b_2 &\cdots & b_n \!\!\!\\
	& \!\!\!c_1 & c_2 &\cdots & c_n \!\!\!
	\end{array} \right] = 	
	\left[ \begin{array}{ccccc}
	& \!\!\!\bar{a}_1 & \bar{a}_2 &\cdots & \bar{a}_n \!\!\! \\
	\!\!\! b_0 & \!\!\! 1 & 1 &\cdots & 1 \!\!\!
	\end{array} \right] ,	\qquad \bar{a}_k = c_1^{-1} b_1 \dots c_{k-2}^{-1} b_{k-2} c_{k-1}^{-1} a_k b_k^{-1} c_k \dots  b_1^{-1} c_1 ,
	\end{equation*}
	while (ii) the sequence	
	\begin{equation*}
	\gamma_k =  a_1 b_1^{-1} \dots  b_{k-1}^{-1} a_k b_k^{-1}, \qquad
	\delta_k = a_k^{-1} b_{k-1}^{-1} a_{k-1}^{-1} \dots b_1 a_1^{-1} 
	\end{equation*} removes both the $a$- and $b$-symbols 
	\begin{equation*}
	\left[ \begin{array}{ccccc}
	& \!\!\!a_1 & a_2 &\cdots & a_n \!\!\! \\
	\!\!\! b_0 & \!\!\!b_1 & b_2 &\cdots & b_n \!\!\!\\
	& \!\!\!c_1 & c_2 &\cdots & c_n \!\!\!
	\end{array} \right] = 	
	\left[ \begin{array}{ccccc}
	\!\!\! b_0 & \!\!\!  1 &  1 &\cdots & 1 \!\!\! \\
	& \!\!\!\bar{c}_1 & \bar{c}_2 &\cdots & \bar{c}_n  \!\!\! , 
	\end{array} \right]  , \qquad 	\bar{c}_k = a_1 b_1^{-1} \dots  a_k b_k^{-1} c_k  a_{k-1}^{-1} b_{k-2} a_{k-2}^{-1} \dots b_1 a_1^{-1} .
	\end{equation*}
\end{Cor}

\section{Continued fractions and quasideterminants}
\label{sec:CF-Q}
\begin{Def}{\cite{Quasideterminants-GR1}}
Given square matrix $X=(x_{ij})_{i,j=1,\dots,n}$ with formal entries $x_{ij}$. In the free division ring generated by the set $\{ x_{ij}\}_{i,j=1,\dots,n}$ consider the formal inverse matrix $Y=X^{-1}= (y_{ij})_{i,j=1,\dots,n}$ to $X$.
The $(i,j)$th quasideterminant $|X|_{ij}$ of $X$ is the inverse $(y_{ji})^{-1}$ of the $(j,i)$th element of $Y$.
\end{Def}
Quasideterminants can be computed using the following recurrence relation~\cite{Quasideterminants-GR1}. For $n\geq 2$ let $X^{ij}$ be the square matrix obtained from $X$ by deleting the $i$th row and the $j$th column (with index $i/j$ skipped from the row/column enumeration), then
\begin{equation} \label{eq:QD-exp}
|X|_{ij} = x_{ij} - \sum_{\substack{ i^\prime \neq i \\ j^\prime \neq j }} x_{i j^\prime} (|X^{ij}|_{i^\prime j^\prime })^{-1} x_{i^\prime j}.
\end{equation}
\begin{Rem}
	One consider quasideterminants of matrices $X$ with some entries vanishing, provided inver\-ti\-bility of the matrix $Y$ and of its elements in question.   
\end{Rem}
\begin{Ex}
	The $(1,1)$ quasideterminant of the matrix
	 \begin{equation*}
	M_2 = \left( \begin{array}{rrr}
	b_0 & a_1 & 0 \\
	-c_{1} & b_1 & a_2 \\
	0 & -c_{2} & b_2
	\end{array}
	\right)
	\end{equation*}
	equals $C_2  = b_0 + a_1 ( b_1 + a_2 b_2^{-1}c_2)^{-1}c_1$.	
\end{Ex}
Non-commutative continued fractions of more involved form then those considered in the present paper were studied in relation to quasideterminants by Gelfand and Retakh  in \cite{Quasideterminants-GR1}. Our double-sided continued fractions are obtained from  their approach by considering tridiagonal matrices. Using the expansion formula \eqref{eq:QD-exp}, one can prove what follows.
\begin{Prop}
	The generalized continued fraction \eqref{eq:CF} equals $(1,1)$ quasideterminant of the tridiagonal matrix
	\begin{equation} \label{eq:QD-CF-gen}
	M_n =
	\left( \begin{array}{ccccc}
	b_0 & a_1 & 0 & \cdots & 0\\
	-c_1 & b_1 & a_2 &\cdots & 0\\
	\vdots & \ddots &\ddots &\ddots & \vdots \\
	0 &\cdots &-c_{n-1}&b_{n-1}&a_n \\
	0 & \cdots & 0 & - c_n & b_n
	\end{array}
	\right)
	\end{equation}
\end{Prop}

\begin{Cor} \label{cor:Mn-dec}
	Equations \eqref{eq:CF-system-XY-a}-\eqref{eq:CF-system-XY-b} can be rewritten in the matrix form as the following decomposition
	\begin{equation*} 
	M_n = 
	\left( \begin{array}{ccccc}
	1 & X_1 & 0 & \cdots & 0\\
	0 & 1 & X_2 &\cdots & 0\\
	\vdots & \ddots &\ddots &\ddots & \vdots \\
	0 &\cdots & 0 & 1 &X_n \\
	0 & \cdots & 0 & 0 & 1
	\end{array}
	\right)
	\left( \begin{array}{ccccc}
	Y_0 & 0 & 0 & \cdots & 0\\
	-c_1 & Y_1 & 0 &\cdots & 0\\
	\vdots & \ddots &\ddots &\ddots & \vdots \\
	0 &\cdots &-c_{n-1}&Y_{n-1}& 0 \\
	0 & \cdots & 0 & - c_n & Y_n
	\end{array}
	\right)
	\end{equation*}	
	
\end{Cor}

\section{The $qd$-algorithm for non-commutative double-sided continued fractions}
\label{sec:LR-QD}
In this Section we give the corresponding analog of the well known $LR$ algorithm (and thus also the $qd$-algorithm) by Rutishauser~\cite{Rutishauser-qd,Rutishauser-LR}. 
Let us consider the following (similar to that described by Corollary~\ref{cor:Mn-dec}) factorization of the matrix $M_n$,  in terms of lower-triangular matrix  $L_n$ and upper-triangular matrix 
$R_n$
\begin{equation*} 
M_n = L_n R_n =
\left( \begin{array}{ccccc}
1 & 0 & 0 & \cdots & 0\\
-Z_1 & 1 & 0 &\cdots & 0\\
\vdots & \ddots &\ddots &\ddots & \vdots \\
0 &\cdots & -Z_{n-1} & 1 & 0 \\
0 & \cdots & 0 & -Z_n & 1
\end{array}
\right)
\left( \begin{array}{ccccc}
Y_0 & a_1 & 0 & \cdots & 0\\
0 & Y_1 & a_2 &\cdots & 0\\
\vdots & \ddots &\ddots &\ddots & \vdots \\
0 &\cdots & 0 &Y_{n-1}& a_n \\
0 & \cdots & 0 & 0 & Y_n
\end{array}
\right), 
\end{equation*}	
which gives the identification
\begin{align} \label{eq:CF-LR-0-b}
b_k  & = Y_k -  Z_k a_k, &  Z_0 = 0, \qquad k=0,1, \dots ,n ,\\ \label{eq:CF-LR-0-c}
c_k  & = Z_k Y_{k-1} ,  &  k=1,2, \dots , n .
\end{align}
Define matrix $M_n^\prime$ by the opposite decomposition $M_n^\prime = R_n L_n $, i.e.
\begin{equation*} 
M_n^\prime 
=  \left( \begin{array}{ccccc}
Y_0 - a_1 Z_1   & a_1  & 0 & \cdots & 0\\
- Y_1  Z_1  & Y_1  - a_2 Z_2   & a_2  & \cdots & 0\\
\vdots & \ddots &\ddots &\ddots & \vdots \\
0 &\cdots &- Y_{n-1}  Z_{n-1} &Y_{n-1}  - a_n Z_n  & a_n  \\
0 & \cdots & 0 & - Y_n  Z_n  & Y_n 
\end{array}
\right) . 
\end{equation*}	
Such transformation preserves the tridiagonal form of the matrix and defines therefore its new coefficients  $b^\prime_k$ and $c^\prime_k$ (notice that $a^\prime_k = a_k$). By analogs of equations \eqref{eq:CF-LR-0-b}-\eqref{eq:CF-LR-0-c} we have also new coefficients $Y_k^\prime$ and $Z_k^\prime$. 

By repeating the transformation, we obtain the dynamical system with parameters $a_1, \dots , a_n$
\begin{align} \label{eq:CF-LR-1}
Y^{(m+1)}_k -  Z_k^{(m+1)}a_k & = Y_k^{(m)} - a_{k+1} Z_{k+1}^{(m)} ,\\
\label{eq:CF-LR-2}
Z_k^{(m+1)} Y_{k-1}^{(m+1)} & = Y_k^{(m)} Z_k^{(m)},
\end{align}
where, by definition, $Z^{(m)}_0 = Z^{(m)}_{n+1} = 0$ for all $m=0,1,2,\dots$. This is the corresponding version of the $LR$ algorithm, where the initial data are $Y^{(0)}_0$, $Y^{(0)}_1, Z^{(0)}_1$, \dots , $Y^{(0)}_n, Z^{(0)}_n$.
\begin{Rem}
	Equations \eqref{eq:CF-LR-1}-\eqref{eq:CF-LR-2} can be considered also as a non-commutative discrete-time Toda lattice equation; see~\cite{Hirota-1983} or \cite{IDS} for the commutative case.
\end{Rem}

The same dynamical system \eqref{eq:CF-LR-1}-\eqref{eq:CF-LR-2} but with initial data $Y^{(m)}_0$, where still $Z^{(m)}_0=0$ for all $m\geq 0$, gives the corresponding version of the $qd$-algorithm.
\begin{Rem}
	The non-commutative version of $qd$-algorithm for one-sided continued fractions by Wynn~\cite{Wynn} is obtained from the above equations by putting $a_k = -1$ for all $k$, and with identification $q^{(m)}_k = Y^{(m)}_{k-1}$, $e^{(m)}_k = Z^{(m)}_{k}$.
\end{Rem}

\section{Periodic continued fractions} \label{sec:CF-per}
\begin{Def} \label{def:CF-sP}
	Strictly periodic continued fraction with period $P> 0$ is infinite continued fraction subject to periodicity condition $a_{k+P+1}=a_{k+1}$, $b_{k+P}=b_k$, $c_{k+P+1}=c_{k+1}$, for all $k\geq 0$, i.e. 
	\begin{equation} \label{eq:CF-sP-Y}
	Y = \left[ \begin{array}{ccccc}
	& \!\!\!a_1 & \cdots & a_{P-1}& a_P \!\!\! \\
	\!\!\! b_0 & \!\!\!b_1 & \cdots & b_{P-1} & Y \!\!\!\\
	& \!\!\!c_1 & \cdots &c_{P-1} & c_P \!\!\!
	\end{array} \right] =
	b_0 + a_1 ( b_1 + \dots + a_{P-1} ( b_{P-1} + a_P Y^{-1}c_P )^{-1}c_{P-1} \dots )^{-1}  c_1 .	\end{equation}
\end{Def}
\begin{Rem}
	The periodicity condition can be transferred to infinite periodic tridiagonal matrices giving the qusideterminantal description of non-commutative strictly periodic continued fractions. 
		Equations \eqref{eq:CF-system-XY-a}-\eqref{eq:CF-system-XY-b} in the periodic case can be rewritten however in the finite matrix form (here $Y_0 = Y$)
		\begin{multline*}  \!\!\!\!\!
		\left( \!\!\! \begin{array}{ccccc}
		b_0 & a_1 &  & \cdots & -c_P\\
		-c_1 & b_1 &  &\cdots & 0\\
		\vdots & \ddots &\ddots &\ddots & \vdots \\
		0 &\cdots &&b_{P-2}&a_{P-1} \\
		a_P & 0 & \cdots  & - c_{P-1} & b_{P-1} 
		\end{array} \!\!\!
		\right) = 
		\left( \!\!\! \begin{array}{ccccc}
		1 & X_1 &  & \cdots & 0\\
		0 & 1 &  &\cdots & \\
		\vdots & &\ddots &\ddots & \vdots \\
		0 &\cdots &  & 1 &X_{P-1} \\
		X_0 &  0 & \cdots & 0 & 1
		\end{array} \!\!\!
		\right) \!\!
		\left( \!\!\! \begin{array}{ccccc}
		Y_0 & 0 &  \cdots & 0 & -c_P\\
		-c_1 & Y_1 & \ddots & \cdots & 0\\
		\vdots & \ddots& \ddots &  & \vdots \\
		& &&Y_{P-2}& 0 \\
		0 & \cdots& &  - c_{P-1} & Y_{P-1}
		\end{array} \!\!\!
		\right) . 
		\end{multline*}	
\end{Rem}

\begin{Prop} \label{prop:CF-sP-Y}
	The strictly periodic continued fraction given by \eqref{eq:CF-sP-Y} satisfies the following second degree equation
	\begin{equation} \label{eq:CF-2-order-eq}
	Y A\, Y + B \,  Y + Y C + D = 0,
	\end{equation}
	where
	\begin{equation} \label{eq:CF-2-order-eq-ABCD}
	A = B_{P-1}c_P^{-1} , \qquad B = - A_{P-1} c_P^{-1} , \qquad C =  B_{P-2} c_{P-1}^{-1} a_P , \qquad D= - A_{P-2}  c_{P-1}^{-1} a_P .
	\end{equation}
\end{Prop}
\begin{proof}
In the strictly periodic continued fraction $Y$ the coefficient $b_P$ is replaced by $Y$, therefore
	\begin{equation*}
	Y = A_P(b_P \to Y) B_{P}^{-1} (b_P \to Y) = 
	\left( A_{P-2}c_{P-1}^{-1} a_P + A_{P-1}c_P^{-1} Y \right) 
	\left( B_{P-2}c_{P-1}^{-1} a_P + B_{P-1}c_P^{-1} Y \right)^{-1}.
	\end{equation*}
\end{proof}
\begin{Rem}
	Solutions of equation \eqref{eq:CF-2-order-eq} with operator coefficients in terms of  strictly $2$-periodic continued fractions, and related convergence questions were discussed in \cite{BusbyFair}.
\end{Rem}
\begin{Cor}
	Because of formula \eqref{eq:CF-2m2-AB} equation \eqref{eq:CF-sP-Y} of the strictly periodic continued fraction $Y$ can be formulated as the following  eigenvalue problem 
\begin{equation} 
\left( \begin{array}{c}
1 \\ Y \end{array} \right) \lambda = 
\left( \begin{array}{cc}
B_{P-2}c_{P-1}^{-1} a_{P}  & B_{P-1}c_P^{-1} \\ A_{P-2}c_{P-1}^{-1} a_{P} & A_{P-1}c_P^{-1} \end{array} \right) \left( \begin{array}{c}
1 \\  Y \end{array} \right) .
\end{equation}	
Elimination of the eigenvalue $\lambda$ gives equation \eqref{eq:CF-2-order-eq}.
\end{Cor}
\begin{Rem}
	From the other side, formula \eqref{eq:CF-2m2-AB-tilde} implies another  eigenvalue problem 
	\begin{equation} \tilde{\lambda}
	\left( \begin{array}{cc} 1 \;, & \!\!\!Y \end{array} \right) =
	\left( \begin{array}{cc} 1 \; , & \!\!\! Y \end{array} \right)
	\left( \begin{array}{cc}
	c_P a_{P-1}^{-1} \tilde{B}_{P-2}  & 	c_P a_{P-1}^{-1} \tilde{A}_{P-2} \\ 
	a_P^{-1} \tilde{B}_{P-1}   & a_P^{-1} \tilde{A}_{P-1}   \end{array} \right) ,
	\end{equation}	
	which leads to equation \eqref{eq:CF-2-order-eq} but with coefficients
	\begin{equation}
	A = a_P^{-1} \tilde{B}_{P-1} , \qquad B = c_P a_{P-1}^{-1} \tilde{B}_{P-2}  , \qquad 
	C =  - a_P^{-1} \tilde{A}_{P-1} , \qquad D= - c_P a_{P-1}^{-1} \tilde{A}_{P-2}  .
	\end{equation}
\end{Rem}

\begin{Def}
	Periodic (or ultimately periodic) continued fraction is infinite continued fraction subject to periodicity condition $a_{k+P+1}=a_{k+1}$, $b_{k+P}=b_k$, $c_{k+P+1}=c_{k+1}$ for all $k\geq K\geq 0$, i.e.
	\begin{equation} \label{eq:CF-P-X}
	X = \left[ \begin{array}{ccccc}
	& \!\!\!a_1 & \cdots &a_{K-1} & a_K \!\!\! \\
	\!\!\! b_0 & \!\!\!b_1 & \cdots & b_{K-1} & Y \!\!\!\\
	& \!\!\!c_1 & \cdots & c_{K-1} & c_K \!\!\!
	\end{array} \right]	\end{equation}
	where $Y$ is strictly periodic continued fraction
	\begin{equation} \label{eq:CF-P-sP-Y}
	Y = \left[ \begin{array}{ccccc}
	& \!\!\!a_{K+1} & \cdots & a_{K+P-1} & a_{K+P} \!\!\! \\
	\!\!\! b_K & \!\!\!b_{K+1} &  \cdots & b_{K+P-1} &Y \!\!\!\\
	& \!\!\!c_{K+1} & \cdots & c_{K+P-1} & c_{K+P} \!\!\!
	\end{array} \right].	\end{equation}
\end{Def}
\begin{Rem}
The periodicity condition is not preserved under equivalence transformations~\eqref{eq:CF-equiv}. In particular, in the simplest case of strictly periodic double-sided continued fraction \eqref{eq:CF-sP-Y} of even period $P$, the coefficients of its reduced form, given by Proposition~\ref{prop:CF-red}, satisfy the following quasi-periodicity conditions
\begin{gather*}
\bar{b}_{mP+2i} = \alpha^m \bar{b}_{2i} \beta^{m}, \qquad  \bar{b}_{mP+2i+1} = \beta^{-m} \bar{b}_{2i+1} \alpha^{-m} ,\qquad 0\leq 2i, 2i+1 < P, \\
\text{where} \qquad \alpha = a_1 c_2^{-1} \dots a_{P-1}c_P^{-1}, \qquad \beta = a_P^{-1} c_{P-1} \dots a_2^{-1} c_1 .
\end{gather*}	
\end{Rem}
The following result is a non-commutative analogue of the theorem of Euler~\cite{Euler}, which states that periodic continued fractions satisfy quadratic equations.
\begin{Prop}
	The periodic continued fraction $X$ given by \eqref{eq:CF-P-X}-\eqref{eq:CF-P-sP-Y} satisfies the second degree equation of the form \eqref{eq:CF-2-order-eq} with coefficients
	\begin{align*}
	A  = & \quad B_{K+P-1} c_{K+P}^{-1}  c_K a_{K-1}^{-1} \tilde{B}_{K-2} 
	 - B_{K+P-2} c_{K+P-1}^{-1} a_{K+P} a_K^{-1} \tilde{B}_{K-1} ,
	\\
	B  = & - A_{K+P-1} c_{K+P}^{-1}  
	c_K a_{K-1}^{-1} \tilde{B}_{K-2}  
	 + A_{K+P-2}  c_{K+P-1}^{-1} a_{K+P} a_K^{-1} \tilde{B}_{K-1} ,
	\\
	C  = & - B_{K +P -1}  c_{K+P}^{-1}  c_K a_{K-1}^{-1}  \tilde{A}_{K-2} 
	+ B_{K+P-2}  c_{K+P-1}^{-1} a_{K+P} a_K^{-1} \tilde{A}_{K-1} ,
	\\
	D = & \quad A_{K+P-1} c_{K+P}^{-1}  c_K a_{K-1}^{-1}  \tilde{A}_{K-2}  - 
	A_{K+P-2} c_{K+P-1}^{-1} a_{K+P} a_K^{-1} \tilde{A}_{K-1} .
	\end{align*}	
\end{Prop}
\begin{proof}
	By Proposition~\ref{prop:CF-sP-Y} the strictly periodic part $Y$ given by \eqref{eq:CF-P-sP-Y} satisfies the quadratic equation
	\begin{equation} \label{eq:Y-2-P}
	Y B^K_{P-1}c_{K+P}^{-1}  Y - A^K_{P-1} c_{K+P}^{-1}  Y + Y B^K_{P-2} c_{K+P-1}^{-1} a_{K+P} - A^K_{P-2}  c_{K+P-1}^{-1} a_{K+P} = 0.
	\end{equation}
	By Corollary~\ref{cor-CF-Serret} we have right and left fractional-linear transformations connecting $X$ and $Y$
	\begin{align} \label{eq:Y-X-l}
	Y & = ( X B_{K-1}c_K^{-1} - A_{K-1}c_K^{-1} )^{-1} (A_{K-2} c_{K-1}^{-1} a_{K} - X B_{K-2}c_{K-1}^{-1} a_K), \\ \label{eq:Y-X-r}
	Y & = ( c_K a_{K-1}^{-1} \tilde{A}_{K-2} - c_K a_{K-1}^{-1} \tilde{B}_{K-2})  (a_K^{-1} \tilde{B}_{K-1} X - a_{K}^{-1} \tilde{A}_{K-1})^{-1}.
	\end{align}
	Inserting into equation \eqref{eq:Y-2-P} the representation \eqref{eq:Y-X-l}  for $Y$ on the left, and \eqref{eq:Y-X-r}  for $Y$ on the right, and multiplying by denominators we obtain the following second order equation
	\begin{gather*}
	(A_{K-2}  - X B_{K-2}) c_{K-1}^{-1} a_{K} B^K_{P-1}c_{K+P}^{-1}  
	c_K a_{K-1}^{-1} (\tilde{A}_{K-2} - \tilde{B}_{K-2} X ) + \\
	- ( X B_{K-1}- A_{K-1} ) c_K^{-1}  A^K_{P-1} c_{K+P}^{-1}  c_K a_{K-1}^{-1} (\tilde{A}_{K-2} - \tilde{B}_{K-2} X ) + \\
	+ 
	(A_{K-2}  - X B_{K-2}) c_{K-1}^{-1} a_{K} B^K_{P-2} c_{K+P-1}^{-1} a_{K+P} a_K^{-1} 
	(\tilde{B}_{K-1} X - \tilde{A}_{K-1})+ \\
	- ( X B_{K-1}- A_{K-1}) c_K^{-1}  A^K_{P-2}  c_{K+P-1}^{-1} a_{K+P} a_K^{-1} 
	(\tilde{B}_{K-1} X - \tilde{A}_{K-1})= 0,
	\end{gather*}
	of the form \eqref{eq:CF-2-order-eq}. For example, the first coefficient reads
	\begin{align*}
	A  = & \; B_{K-2} c_{K-1}^{-1} a_{K} B^K_{P-1}c_{K+P}^{-1}  
	c_K a_{K-1}^{-1} \tilde{B}_{K-2}  +
	B_{K-1} c_K^{-1}  A^K_{P-1} c_{K+P}^{-1}  c_K a_{K-1}^{-1}  \tilde{B}_{K-2} + \\
	& - B_{K-2} c_{K-1}^{-1} a_{K} B^K_{P-2} c_{K+P-1}^{-1} a_{K+P} a_K^{-1} 
	\tilde{B}_{K-1} - B_{K-1} c_K^{-1}  A^K_{P-2}  c_{K+P-1}^{-1} a_{K+P} a_K^{-1} \tilde{B}_{K-1} ,
	\end{align*}
	but it can be simplified using Corollary~\ref{cor:CF-splitting}.
\end{proof}

\section{A weak analogue of the Galois theorem for continued fractions} \label{sec:Galois}
The final Section is devoted to a non-commutative analogue of the theorem published by  Galois~\cite{Newmann-Galois} at the age of 17, which states that when $y$ is a strictly periodic continued fraction then the second solution $y^\prime$ of its quadratic equation is such that $-(y^\prime)^{-1}$ is given by strictly periodic continued fraction with the same period and coefficients but in the reversed order
\begin{equation*}
y =[b_0, b_1, \dots , b_{P-1},y] \qquad \Rightarrow \qquad 
-(y^\prime)^{-1}= [b_{P-1}, \dots b_1, b_0, -(y^\prime)^{-1}].
\end{equation*}

To simplify calculations let us replace strictly periodic double-sided continued fraction \eqref{eq:CF-sP-Y} by the left sided one (notice that this is not an equivalence transformation)
\begin{equation} \label{eq:CF-sP-left}
\bar{Y} = \left[ \begin{array}{ccccc}
& \!\!\!\bar{a}_1 & \cdots & \bar{a}_{P-1}& \bar{a}_P \!\!\! \\
\!\!\! \bar{b}_0 & \!\!\!\bar{b}_1 & \cdots & \bar{b}_{P-1} & \bar{Y} \!\!\!
\end{array} \right] , \qquad \bar{Y} = Y c_1^{-1}, \qquad \bar{b}_k = b_k c_{k+1}^{-1}, \qquad \bar{a}_{k} = a_k c_{k+1}^{-1},	
\end{equation}
where indices are considered modulo $P$. As in Section~\ref{sec:CF-2x2} we replace the process of calculation of the strictly periodic continued fraction \eqref{eq:CF-sP-left} by solving $\bar{Y}_0 = \bar{Y}$ from the system
\begin{align} \label{eq:CF-XY-bar-1}
\bar{X}_k \bar{Y}_k & = \bar{a}_k, \\ \label{eq:CF-XY-bar-2}
\bar{Y}_k - \bar{X}_{k+1} & = \bar{b}_{k},
\end{align}
equations \eqref{eq:CF-system-XY-a}-\eqref{eq:CF-system-XY-b} in reduction \eqref{eq:CF-sP-left}.
This time, however, we are looking for another solution $\bar{Y}_0^\prime=Y^\prime c_1^{-1}$, different from 
$\bar{Y}_0$, of the related equation \eqref{eq:CF-2-order-eq}. Notice that this implies that all $\bar{X}_k^\prime$ and $\bar{Y}_k^\prime$,
calculated successively from the system
\begin{align} \label{eq:CF-XY-KP-1}
\bar{X}_k^\prime \bar{Y}_k^\prime & = \bar{X}_k \bar{Y}_k , \\ \label{eq:CF-XY-KP-2}
\bar{Y}_k^\prime - \bar{X}_{k+1}^\prime & = \bar{Y}_k - \bar{X}_{k+1}  ,
\end{align}
differ from the corresponding $\bar{X}_k$ and $\bar{Y}_k$. 
\begin{Rem}
	The system \eqref{eq:CF-XY-KP-1}-\eqref{eq:CF-XY-KP-2} up to a change of sign coincides with equations \eqref{eq:xy-k} defining the companion to the KP map, see Section~\ref{sec:comp-KP}, and was our original motivation to investigate the subject of non-commutative continued fractions.
\end{Rem}
The next result provides the corresponding analogs of equations~\eqref{eq:y-h} and \eqref{eq:hk-series}.
\begin{Lem}
Under the assumption $\bar{Y}_0^\prime\neq \bar{Y}_0$ the expression $(1 - \bar{Y}_0^\prime \bar{Y}_0^{-1})^{-1}$ is the following ultimately periodic continued fraction
\begin{equation} \label{eq:CF-H}
(1 - \bar{Y}_0^\prime \bar{Y}_0^{-1})^{-1} = 
\left[ \begin{array}{cccccc}
\!\!\! 1  & \!\!\! -\bar{Y}_{P-1} & \bar{X}_{P-1} -\bar{Y}_{P-2}  & \bar{X}_{P-2} -\bar{Y}_{P-3} & \bar{X}_{P-3} -\bar{Y}_{P-4} &  \cdots \!\!\! \\
& \!\!\!\bar{X}_{0} & \bar{X}_{P-1} \bar{Y}_{P-1} & \bar{X}_{P-2} \bar{Y}_{P-2}  &  
 \bar{X}_{P-3} \bar{Y}_{P-3} & \cdots \!\!\! 
\end{array} \right]
\end{equation}
given in terms of the initial solution of the system \eqref{eq:CF-XY-bar-1}-\eqref{eq:CF-XY-bar-2}. Moreover, the explicit form of the expression can be given as the following series
\begin{align} \label{eq:CF-H-series}
(1 - \bar{Y}_0^\prime \bar{Y}_0^{-1})^{-1}  &=  1 + \sum_{k=1}^\infty (-1)^{k}
\bar{Y}_{P-1}^{-1} \dots \bar{Y}_{P-k}^{-1} 
\bar{X}_{P-k+1} \bar{X}_{P-k} \dots \bar{X}_{0}  = \\ \nonumber
 = 
1 - \bar{Y}_{P-1}^{-1} \bar{X}_{0}  & + 
 \bar{Y}_{P-1}^{-1}\bar{Y}_{P-2}^{-1} \bar{X}_{P-1} \bar{X}_{0}  
 - \bar{Y}_{P-1}^{-1}\bar{Y}_{P-2}^{-1} \bar{Y}_{P-3}^{-1} \bar{X}_{P-2}\bar{X}_{P-1} \bar{X}_{0} 
 + \dots \quad .
\end{align}
\end{Lem}
\begin{proof}
	Equations for $k=0$ and $k=P-1$ of the system \eqref{eq:CF-XY-KP-1}-\eqref{eq:CF-XY-KP-2} give
	\begin{equation*}
(1 - \bar{Y}_0^\prime \bar{Y}_0^{-1} )^{-1} =  1 + ( \bar{Y}_{P-1}^{\prime} - \bar{Y}_{P-1} )\bar{X}_{0}.
	\end{equation*}
	Then successive application of the recurrence, established by \eqref{eq:CF-XY-KP-1}-\eqref{eq:CF-XY-KP-2}, 
	\begin{equation*}
	\bar{Y}_k^\prime = ( \bar{X}_{k} -\bar{Y}_{k-1} + \bar{Y}_{k-1}^\prime)^{-1} \bar{X}_{k} \bar{Y}_{k}
	\end{equation*}
	gives the continued fraction form \eqref{eq:CF-H}.
	
	By induction, using equations \eqref{eq:CF-rec-A-rev}-\eqref{eq:CF-rec-B-rev}, one can give the following explicit form of the nominators and denominators of the left simple fraction decompositions \eqref{eq:CF-BA} of convergents of the continued fraction~\eqref{eq:CF-H}
	\begin{align} \label{eq:CF-Gal-A}
	\tilde{A}_k &  = \bar{X}_{P-k+1}\bar{X}_{P-k+2}\dots \bar{X}_{0} -
	\bar{Y}_{P-k}\bar{X}_{P-k+2}\dots \bar{X}_{0} + \dots \\
	\nonumber
	 \dots & + (-1)^{k-1} \bar{Y}_{P-k}\dots \bar{Y}_{P-2}\bar{X}_{0} + (-1)^k
	\bar{Y}_{P-k}\dots \bar{Y}_{P-2}\bar{Y}_{P-1} , \\
	\label{eq:CF-Gal-B}
	\tilde{B}_k &= (-1)^k
	\bar{Y}_{P-k}\dots \bar{Y}_{P-2}\bar{Y}_{P-1} ,	
	\end{align}
	which gives the series \eqref{eq:CF-H-series}.
\end{proof}
\begin{Rem}
The series \eqref{eq:CF-H-series} can be also obtained from the Euler--Minding type expansion \eqref{eq:CF-exp-rev} and expression	\eqref{eq:CF-Gal-B}.
\end{Rem}
From the continued fraction representation \eqref{eq:CF-H} one can find the continued fraction representation of $\bar{Y}^\prime_0 = Y^\prime c_1^{-1}$
\begin{equation*}
\bar{Y}^\prime_0 =   
\left[ \begin{array}{ccccc}
\!\!\! 0  & \!\!\! \bar{X}_{0} -\bar{Y}_{P-1} & \bar{X}_{P-1} -\bar{Y}_{P-2}  & \bar{X}_{P-2} -\bar{Y}_{P-3} &   \cdots \!\!\! \\
& \!\!\!\bar{X}_{0} \bar{Y}_{0} & \bar{X}_{P-1} \bar{Y}_{P-1} & \bar{X}_{P-2} \bar{Y}_{P-2}  &  \cdots \!\!\! 
\end{array} \right] =     
\left[ \begin{array}{ccccc}
\!\!\! 0  & \!\!\!  -\bar{b}_{P-1} &  -\bar{b}_{P-2}  &  -\bar{b}_{P-3} &   \cdots \!\!\! \\
& \!\!\!\bar{a}_{P} & \bar{a}_{P-1} & \bar{a}_{P-2} &  \cdots \!\!\! 
\end{array} \right] .
\end{equation*}
Going back to the initial coefficients of the strictly periodic continued fraction \eqref{eq:CF-sP-Y} we obtain that 
$Z= -a_P (Y^\prime)^{-1} c_P$ is strictly periodic continued fraction of the form
\begin{equation} 
Z= -a_P (Y^\prime)^{-1} c_P = \left[ \begin{array}{ccccc}
& \!\!\!c_{P-1} & \cdots & c_1 & c_P \!\!\! \\
\!\!\! b_{P-1}  & \!\!\!  b_{P-2}  & \cdots &
b_{0}    & Z \!\!\!\\
& \!\!\!a_{P-1} & \cdots  &a_{1} & a_P \!\!\!
\end{array} \right] ,	\end{equation}	
which implies the following final result.
\begin{Prop}
	When $Y$ is the strictly periodic continued fraction \eqref{eq:CF-sP-Y} then
	$Y^\prime$ given by the following strictly periodic continued fraction
	\begin{equation} 
	-(Y^\prime)^{-1} = \left[ \begin{array}{ccccc}
	& \!\!\!a_P^{-1} & \cdots & a_2^{-1} & a_1^{-1} \!\!\! \\
	\!\!\! a_P^{-1} b_{P-1} c_P^{-1} & \!\!\!a_{P-1}^{-1} b_{P-2} c_{P-1}^{-1} & \cdots &
	a_{1}^{-1} b_{0} c_{1}^{-1}  & -(Y^\prime)^{-1} \!\!\!\\
	& \!\!\!c_P^{-1} & \cdots  &c_2^{-1} & c_1^{-1} \!\!\!
	\end{array} \right] .	\end{equation}	
	is another solution of the corresponding equation \eqref{eq:CF-2-order-eq}-\eqref{eq:CF-2-order-eq-ABCD}.
\end{Prop}

\begin{Cor}
	On the level of strictly periodic continued fractions with unital $a$- and $c$-coefficients 
	\begin{equation*} 
	Y = b_0 + ( b_1 +  ( b_2 + \dots +  ( b_{P-1} +  Y^{-1})^{-1}  \dots )^{-1}  )^{-1}
	= [b_0,b_1 , \dots , b_{P-1},Y],
	\end{equation*} 
	we have
	\begin{equation*} 
	-(Y^\prime)^{-1} = b_{P-1} + ( b_{P-2} +  ( b_{P-3} + \dots +  ( b_{0} -  Y^\prime)^{-1}  \dots )^{-1}  )^{-1}=
	[b_{P-1} , b_{P-2}, \dots , b_0, -(Y^\prime)^{-1}].	
	\end{equation*} 	
\end{Cor}

\section{Conclusion}
Motivated by the theory of non-commutative integrable discrete systems we studied non-commutative continued fractions in their generic double-sided form. We presented the corresponding versions of the most pertinent results known from the classical theory. We think that the non-commutative perspective sheds new light on old well known facts and broadens applications of the theory.  
However, in some cases we gave only "weak" version of a result, meaning that we show only implication in one direction. For example, we do not have full version of the Lagrange theorem on quadratic irrationals but we present only Euler's part of the theorem. Notice however, that the simplest case of power series with single indeterminate~\cite{Schmidt,Khrushchev}, which in our approach corresponds to the free group with one generator, suggests that the reverse implication is valid only under certain additional conditions. This aspect is presently under investigation.

The present paper shows that it is possible to transfer to the fully non-commutative level several results on relation of integrable systems and continued fractions. The results are valid for arbitrary division ring. An interesting question is what happens on the intermediate level between full non-commutativity and commutativity --- what commutation relations can be imposed without destroying integrability? What additional properties appear when such integrable constraints are imposed? Examples presented in \cite{Sergeev-Q2+1,Doliwa-Sergeev,Doliwa-Kashaev} show intriguing link between the ultra-locality principle and the Weyl commutation relations. In general, the rings of fractions of the standard quantum algebras may be interesting objects to study in this context.

\section*{Acknowledgments}
The author would like to thank anonymous referees for careful reading of the manuscript and for their constructive remarks which helped to improve presentation. 
The paper was supported by National Science Centre, Poland, under grant 2015/19/B/ST2/03575 \emph{Discrete integrable systems -- theory and applications}. The author would like to thank organizers of the workshop \emph{ISLAND V: Integrable systems, special functions and combinatorics} (23-28 June, Sabhal Mor Ostaig, Isle of Skye) for invitation and support.

\bibliographystyle{amsplain}

\providecommand{\bysame}{\leavevmode\hbox to3em{\hrulefill}\thinspace}

\end{document}